\documentclass[conference]{IEEEtran}
\usepackage[T1]{fontenc}
\usepackage{graphicx}
\usepackage[noadjust]{cite}
\usepackage{mcite}
\usepackage{amsfonts,helvet}
\usepackage{fancyhdr}
\usepackage{threeparttable}
\usepackage{epsf,epsfig}
\usepackage{amsthm}
\usepackage{amsmath}
\usepackage{siunitx}
\usepackage{amssymb}
\usepackage{dsfont}
\usepackage{subfigure}
\usepackage{color}
\usepackage{algorithm}
\usepackage{algpseudocode}
\usepackage{algcompatible}
\usepackage{enumerate}
\usepackage{hyperref}
\usepackage{cancel}
\usepackage{bbm}
\usepackage{dsfont}
\usepackage[subnum]{cases}
\usepackage{adjustbox}
\usepackage{dblfloatfix} 
\usepackage{array}
\usepackage{url}

\newtheorem{corollary}{Corollary}

\newtheorem{proposition}{Proposition}
\newtheorem{remark}{Remark}

\newmuskip\pFqmuskip

\def\cS{\mbox{$\mathcal{S}$}}

\def\argmax{\mathop{\rm \arg\!\max}}

\makeatletter
\newcommand{\thickhline}{%
    \noalign {\ifnum 0=`}\fi \hrule height 1pt
    \futurelet \reserved@a \@xhline
}
\newcolumntype{"}{@{\hskip\tabcolsep\vrule width 1pt\hskip\tabcolsep}}
\makeatother

\IEEEoverridecommandlockouts

\title{
User Scheduling for Millimeter Wave MIMO Communications with Low-Resolution ADCs}
\author{
Jinseok Choi and Brian L. Evans \\
Wireless Networking and Communication Group, The University of Texas at Austin\\
E-mail: jinseokchoi89@utexas.edu, bevans@ece.utexas.edu
}

\begin{document}
%
\maketitle

\begin{abstract}
	To reduce power consumption in the receiver, low-resolution analog-to-digital converters (ADCs) can be a potential solution for millimeter wave (mmWave) systems in which many antennas are likely to be deployed to compensate for the large path loss. 
	In this paper, we investigate uplink user scheduling when a basestation employs low-resolution ADCs with a large number of antennas. 
	Due to quantization error, we show that the channel structure in the beamspace, in addition to the channel magnitude and beamspace orthogonality, plays a key role in maximizing the achievable rates of scheduled users. 
	{\color{black}With the constraint of the equal channel norm, we derive the scheduling criteria that maximize the uplink sum rate for multi-user multiple input multiple output (MIMO) systems with a zero-forcing receiver.}
	Leveraging the derived criteria, we propose an efficient scheduling algorithm for mmWave systems with low-resolution ADCs. 
	Numerical results validate that the proposed algorithm outperforms conventional user scheduling methods in terms of the sum rate.
\end{abstract}

\begin{IEEEkeywords}
Millimeter wave, low-resolution ADCs, user scheduling, channel structure, beamspace.
\end{IEEEkeywords}

\section{Introduction}
\label{sec:intro}


Millimeter wave communication has drawn extensive attention as a promising technology for 5G cellular systems \cite{pi2011introduction,andrews2014will,boccardi2014five}, and evinced its feasibility \cite{rappaport2013millimeter}.
The advantages of remarkably wide bandwidth have encouraged wireless researchers to perform comprehensive studies to resolve practical challenges in the realization of mmWave communications \cite{niu2015survey, heath2016overview}. 
Due to a large signal bandwidth and a high number of bits/sample in mmWave communications, high-resolution ADCs coupled with large antenna arrays demand significant power consumption in the receiver.
Consequently, receiver architectures with low-resolution ADCs \cite{fan2015uplink} have been of interest in recent years.

In mmWave systems, several channel estimation techniques have been proposed for low-resolution ADCs, showing the feasibility of employing low-resolution ADCs by taking the advantage of the large antenna arrays \cite{rusu2015low, mo2014channel, mo2016channel}.
In \cite{rusu2015low}, an adaptive compressed sensing strategy was proposed to recover the sparse mmWave channel and reduced the training overhead when the channel is sparse.
A generalized approximate message-passing algorithm with 1-bit ADCs in \cite{mo2014channel} showed a similar channel estimation performance as maximum-likelihood estimator with full-resolution ADCs in the low and medium SNR regimes by exploiting the sparsity of mmWave channels in the beamspace.
Receiver architectures with resolution-adaptive ADCs were investigated for mmWave systems in \cite{choi2017adc,choi2017resolution}, showing 1-bit improvement in the uplink sum rate in the low-resolution regime.

In another line of research in mmWave communications, user scheduling was investigated in \cite{lee2016randomly} by considering a random beamforming (RBF) method \cite{sharif2005capacity} for uniform random single path channels.
Then, the analysis of user scheduling in mmWave channels was extended to the uniform random multi-path case, and scheduling algorithms were developed based on a beam selection approach \cite{lee2016performance}.
Although the algorithms in \cite{lee2016performance} were proposed for mmWave systems and outperformed the RBF algorithm in \cite{sharif2005capacity}, the algorithms focused on user scheduling without quantization errors.
Accordingly, user scheduling in mmWave systems with low-resolution ADCs is still questionable.

\begin{figure}[!t]\centering
\includegraphics[scale = 0.44]{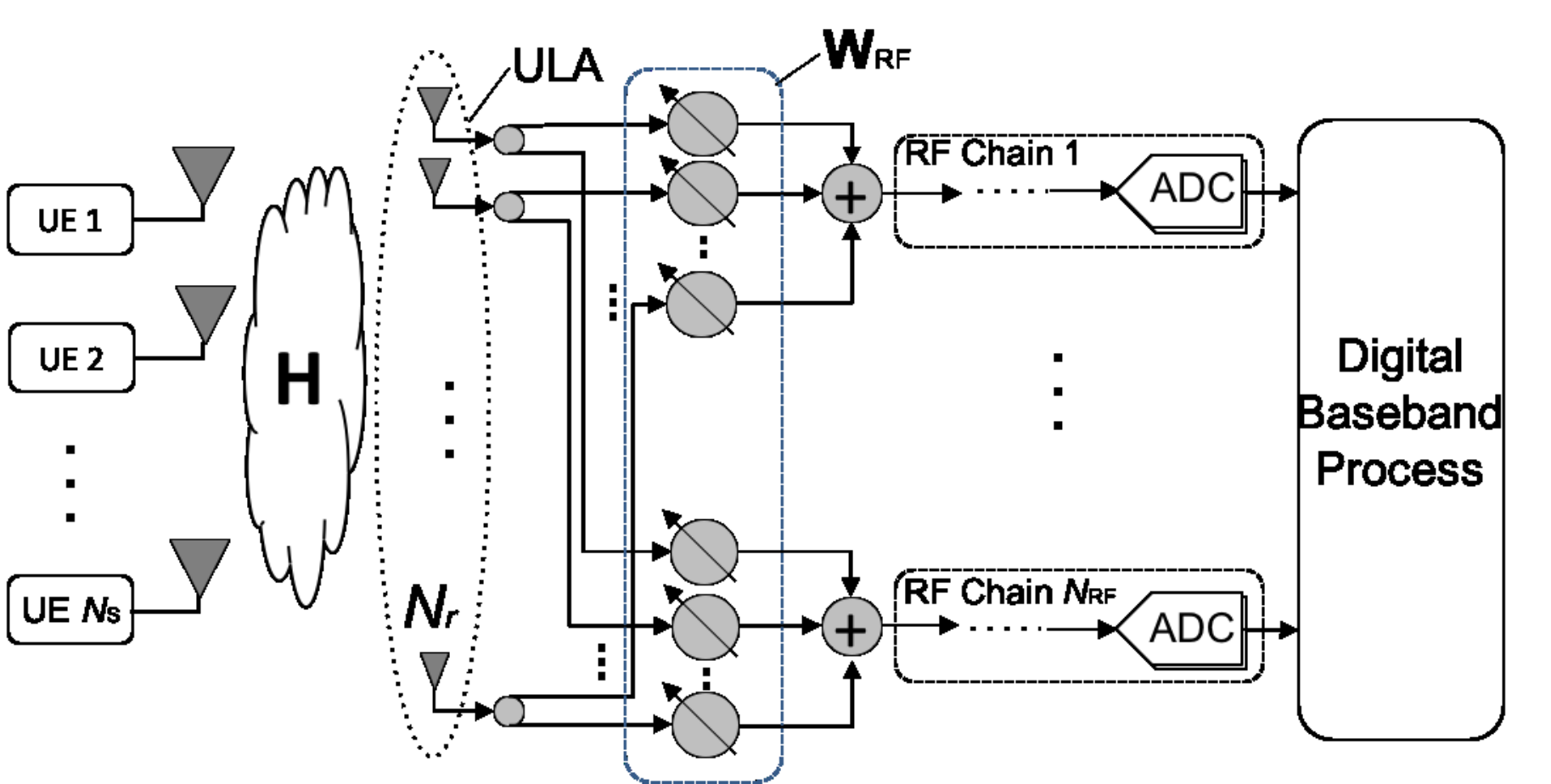}
\caption{A receiver with the large antenna arrays and analog combiner ${\bf W}_{\rm RF}$, followed by low-resolution ADCs in the uplink multi-user system.} 
\label{fig:system}
\end{figure}

In this paper, we investigate uplink user scheduling when a basestation (BS) employs low-resolution ADCs with a large number of antennas in mmWave systems.
Due to the coarse quantization, quantization errors need to be considered in scheduling, which changes the scheduling problem from the conventional problem. 
We derive structural criteria of channels in the beamspace for user scheduling to maximize the uplink sum rate in mmWave MIMO systems with low-resolution ADCs. 
Considering a zero-forcing receiver for multi-user MIMO communications, the derived structural criteria are $(i)$ to have as many channel propagation paths as possible with unique angle of arrivals (AoAs) for each scheduled user and $(ii)$ to have channel power that is evenly spread across beamspace complex gains within each scheduled user channel. 
Leveraging the criteria and mmWave channel sparsity in the beamspace, we propose a user scheduling algorithm with low complexity.
Numerical results demonstrate that the proposed algorithm outperforms previous user scheduling algorithms in terms of the sum rate.

{\color{black} The remainder of the paper is organized as follows.
Section \ref{sec:sys_model} explains the system model.
In Section \ref{sec:user_scheduling}, the user scheduling criteria for maximizing the sum rate with the constraint of the equal channel norm are derived.
In Section \ref{sec:algorithm}, the user scheduling algorithm is proposed by utilizing the derived criteria.
In Section \ref{sec:simulation}, the proposed algorithm is evaluated in sum rate and Section \ref{sec:conclusion} concludes the paper.
}

\section{System Model}
\label{sec:sys_model}

\subsection{Signal and Channel Models}

We consider a single-cell MIMO uplink network where a BS employs uniform linear array (ULA) and analog combiners with $N_r$ antennas as shown in Fig. \ref{fig:system}.
We assume that $N_{u}$ users with a single antenna exist in the cell, and the BS selects $N_s$ users to serve.
$N_{\rm RF}$ RF chains are connected to $N_{\rm RF}$ pairs of ADCs.
The received baseband analog signal ${\bf r} \in \mathbb{C}^{N_r}$ with a narrowband channel assumption is given as
\begin{align}
	\label{eq:r}
	{\bf r} = \sqrt{\rho}{\bf Hs} + {\bf n}
\end{align} 
where $\rho$, ${\bf H}$, ${\bf s}$, and ${\bf n}$ represent the transmit power,  the $N_r \times N_s$ channel matrix, the symbol vector of $N_s$ users, and the additive white Gaussian noise vector, respectively. 
We assume ${\bf n}  \sim \mathcal{CN}(\mathbf{0}, \mathbf{I})$ and Gaussian signaling $\bf s \sim \mathcal{CN}(\mathbf{0}, \mathbf{I})$ where $\mathbf{0}$ is a zero vector and ${\bf I}$ is the identity matrix with proper dimension. 
We consider that the channel $\bf H$ is known at the BS
{\color{black} as previously developed channel estimation techniques \cite{rusu2015low, mo2014channel, mo2016channel,sung2018narrowband} show negligible degradation of the estimation  accuracy compared to the infinite-bit ADC case with few-bit ADCs.}
We assume $N_{\rm RF} = N_r$ as the BS employs low-resolution ADCs.
The BS adopts an $N_r \times N_r$ analog combiner matrix ${\bf W}_{\rm RF}$.
The elements of ${\bf W}_{\rm RF}$ are constrained to have the equal norm of $1/\sqrt{N_r}$; i.e., $[{\bf w}_{{\rm RF},i}{\bf w}^H_{{\rm RF},i}]_{\ell,\ell} = 1/N_r$ where $[\cdot]_{i,i}$ represents the $(\ell,\ell)$th element of $[\cdot]$ and $(\cdot)^H$ denotes the conjugate transpose.
{\color{black} Later, this analog combiner is used to exploit the sparsity of the mmWave channels in user scheduling.}
With the analog beamforming, the received signal \eqref{eq:r} becomes
\begin{align} 
	\label{eq:y}
	{\bf {y}}  = {\bf W}_{\rm RF}^H {\bf r} = \sqrt{\rho}{\bf  W}_{\rm RF}^H{\bf Hs} +  {\bf W}_{\rm RF}^H {\bf n}.
\end{align}

We consider mmWave channels with limited $L$ propagation paths \cite{el2014spatially}.
Then, the $k${th} user channel can be modeled as
\begin{align}
	\label{eq:channel_geo}
	{\bf h}_k = \sqrt{\frac{N_r}{L}}\sum_{\ell = 1}^{L}\omega_{k,\ell} {\bf a}(\theta_{k,\ell})
\end{align}
where $\omega_{k,{\ell}}$ is the complex gain of the $\ell${th} path of the user $k$, and ${\bf a}(\theta_{k,{\ell}})$ is the BS antenna array steering vector corresponding to the azimuth AoA of the $\ell$th path of the user $k$ with $\theta_{k,{\ell}} \in [-\pi/2,\pi/2]$. 
We consider that $\omega_{k,{\ell}}$ is an independent and identically distributed (IID) complex Gaussian random variable as $\omega_{k,{\ell}} \sim \mathcal{CN}(0, 1)$.
The array steering vector ${\bf a}(\theta)$ for ULA is expressed as
$ {\bf a}(\theta) = \frac{1}{\sqrt{N_r}}\Big[1,e^{-j 2\pi\vartheta},e^{-j 4\pi \vartheta},\dots,e^{-j 2(N_r-1)\pi\vartheta}\Big]^\intercal 
$, 
where $\vartheta$ is the normalized spatial angle that is $\vartheta = \frac{d}{\lambda}\sin(\theta)$, $d$ denotes the distance between antenna elements, and $\lambda$ represents the signal wave length.
{\color{black} Assuming uniformly-spaced spatial angles, i.e., $\vartheta_i = \frac{d}{\lambda}\sin(\theta_i) = (i-1)/N_r$, the matrix of array steering vectors
$\mathbf{A}=\big[{\bf a}(\theta_1),\dots,{\bf a}(\theta_{N_r})\big]$
becomes a unitary discrete Fourier transform matrix. }
We adopt the matrix of array steering vectors as the analog combiner ${\bf W}_{\rm RF} = {\bf A}$, and we can rewrite the received signal \eqref{eq:y} as ${\bf y} = \sqrt{\rho} {\bf A}^H{\bf H}{\bf s} + {\bf A}^H{\bf n}$. 
We denote ${\bf H}_{\rm b} = {\bf A}^H{\bf H}$, which is the projection of the channel matrix onto the beamspace.

\subsection{Quantization Model}
\label{subsec:quantization}

After the analog beamforming, 
each real and imaginary component of the complex output $y_i$ is quantized at the $i$th pair of ADCs.
For analytical tractability, we adopt the additive quantization noise model (AQNM) to represent the quantization process in a linear form. 
The AQNM provides reasonable accuracy in low and medium SNR ranges \cite{orhan2015low}.
Accordingly, the quantized signal $\bf y_{\rm q}$ is given as
\begin{align} 
	\label{eq:yq}
	\mathbf{y}_{\rm q}&=\mathcal{ Q}(\mathbf{ y}) = \alpha \sqrt{\rho} {\bf H_{\rm b}}{\bf s} +\alpha {\pmb \eta} + {\bf q}
\end{align} 
where $\mathcal{Q}(\cdot)$ is the element-wise quantizer function for each real and imaginary part, $\alpha$ is the quantization gain which is defined as $\alpha = 1- \beta$ with $\beta = \frac{\mathbb{E}[|{y}_i - {y}_{{\rm q}i}|^2]}{\mathbb{E}[|{y}_i|^2]}$.
Under the assumptions of a scalar minimum mean squared error quantizer and Gaussian signaling, $\beta$ can be approximated as $\beta \approx \frac{\pi\sqrt{3}}{2} 2^{-2b}$ for $b > 5$, and Table 1 in \cite{choi2017resolution} lists the values of $\beta$ for $b \leq 5$.
Note that $b$ is the number of quantization bits for each real and imaginary part of $y_i$.
Since ${\bf A}$ is unitary, the noise  ${\pmb \eta} = {\bf A}^H {\bf n}$ is distributed as ${\pmb \eta} \sim \mathcal{CN}({\bf 0},{\bf I})$.
The quantization noise ${\bf q}$ is an additive noise that is uncorrelated with the quantization input $\bf y$ and follows ${\bf q} \sim \mathcal{CN}({\bf 0},{\bf R}_{\bf qq})$ where the covariance matrix $\mathbf{R}_{\bf qq}$ for a fixed channel realization $ {\bf H}_{\rm b}$ is $\mathbf{R}_{\bf qq}= \alpha\beta\,{\rm diag}(\rho{\bf H_{\rm b}}{\bf H}_{\rm b}^H + {\mathbf{I}})$.
Here, ${\rm diag}(\rho{\bf H_{\rm b}}{\bf H}_{\rm b}^H + \mathbf{I})$ represents the diagonal matrix of the diagonal entries of $\rho{\bf H_{\rm b}}{\bf H}_{\rm b}^H + {\bf I}$.


\section{User Scheduling}
\label{sec:user_scheduling}

In this paper, we focus on employing a zero-forcing combiner ${\bf W}_{\rm zf} = {\bf H}_{\rm b}({\bf H}_{\rm b}^H {\bf H}_{\rm b})^{-1}$ at the BS.
Then, we have
\begin{align}
	{\bf y}^{\rm zf}_{\rm q}
	& = {\bf W}_{\rm zf}^H {\bf y}_{\rm q} = \alpha \sqrt{\rho}{\bf W}_{\rm zf}^H {\bf H}_{\rm b} {\bf s} + \alpha {\bf W}_{\rm zf}^H {\pmb \eta} + {\bf W}_{\rm zf}^H {\bf q},
\end{align}
and the achievable rate for user $k$ is given as
\begin{align}
	\label{eq:rate}
	&\mathcal{R}_k({\bf H}_{\rm b})
	 = \log_2 \left(1+\frac{\alpha^2 \rho}{ {\bf w}_{{\rm zf},k}^H {\bf R}_{{\bf qq}}{\bf w}_{{\rm zf},k}+\alpha^2 \|{\bf w}_{{\rm zf},k}\|^2}\right).
\end{align}
Using the achievable rate in \eqref{eq:rate}, the user scheduling problem can be formulated as
\begin{align}
	\label{eq:problem}
	\mathcal{R}({\bf H}_{\rm b}(\mathcal{S^\star})) =\max_{\mathcal{S} \subset \{1,\dots,N_u\}:|\mathcal{S}| = N_s} \sum_{k \in \mathcal{S}} \mathcal{R}_k({\bf H}_{\rm b}(\mathcal{S}))
\end{align}
where ${\mathcal{S}}$ represents the set of scheduled users and  ${\bf H}_{\rm b}(\mathcal{S})$ is the beamspace channel matrix of the users in $\mathcal{S}$.
Later, we use $\cS(k)$ to indicate the $k$th scheduled user.

\begin{remark}
\label{rm:intuition}
To maximize \eqref{eq:rate}, the aggregated beamspace channel gain at each RF chain $\|[{\bf H}_{\rm b}]_{i,:}\|^2$ needs to be small to reduce the quantization error ${\bf R}_{\bf qq}$ where $[{\bf H}_{\rm b}]_{i,:}$ represents the $i$th row of ${\bf H}_{\rm b}$, in addition to the beamspace channel orthogonality $({\bf h}_{{\rm b},k}\! \perp\! {\bf h}_{{\rm b},k'},\ k \neq k')$ and large beamspace channel gains $\|{\bf h}_{{\rm b},k}\|^2$ for reducing $\|{\bf w}_{{\rm zf},k}\|^2$.
\end{remark}
%

To solve the user scheduling problem in \eqref{eq:problem}, we exploit the finding in Remark \ref{rm:intuition} and derive the structural criteria for channels in the beamspace that maximize the uplink sum rate.
For tractability, we focus on the case where the magnitude of each user channel is equivalent, i.e., $\|{\bf h}_{k}\| = \sqrt{\gamma},\ \forall k$ where $\gamma >0$.
To characterize a channel matrix that maximizes the uplink sum rate, we formulate a problem as follows:
\begin{align}
	\label{eq:opt_prob}
	\mathcal{R}({\bf H}_{\rm b}^\star) =\max_{{\bf H}_{\rm b}\in \mathbb{C}^{N_r \times N_s}}\sum_{k = 1}^{N_s} \mathcal{R}_k({\bf H}_{\rm b}), \quad \text{s.t. } \|{\bf h}_{{\rm b},k}\| = \sqrt{\gamma} \ \ \forall k.
\end{align}
Note that $\|{\bf h}_{{\rm b},k}\| = \sqrt{\gamma}$ in \eqref{eq:opt_prob} is equivalent to $\|{\bf h}_{k}\| = \sqrt{\gamma}$ since the analog combiner ${\bf W}_{\rm RF} = {\bf A}$ is a unitary matrix.


To provide geographical interpretation for the analysis, we adopt the virtual channel representation~\cite{sayeed2002deconstructing} in which each beamspace channel ${\bf h}_{{\rm b},k}$ contains $L$ non-zero complex gains.
The $L$ non-zero elements correspond to the complex gains of $L$ channel propagation paths.
We first analyze the single user selection, followed by the multi-user selection. 

\begin{proposition}
	\label{pr:rate1}
	For single user selection with $\|{\bf h}_{k}\| = \sqrt{\gamma},\ \forall k$, selecting a user who has the following channel characteristics maximizes the uplink achievable rate:
	\begin{enumerate}[(i)]
		\item The largest number of channel propagation paths. 
		\item Equal power spread across the beamspace complex gains.
	\end{enumerate}
\end{proposition}
\begin{proof}
	For a single user, the zero-forcing combiner becomes ${\bf w}_{\rm zf} = {{\bf h}_{\rm b}}/{\|{\bf h}_{\rm b}\|^2}$. 
	Then, the achievable rate \eqref{eq:rate} is given as
	
	\vspace{-1 em}
	\small
	\begin{align}
		\nonumber
		&\mathcal{R}({\bf h}_{\rm b}) 
		= \log_2 \left(1+\frac{\alpha \rho}{\beta\frac{{\bf h}_{\rm b}^H}{\|{\bf h}_{\rm b}\|^2} {\rm diag}\Big(\rho{\bf h}_{\rm b} {\bf h}_{\rm b}^H + \frac{1}{\beta}{\bf I}\Big) \frac{{\bf h}_{\rm b}}{\|{\bf h}_{\rm b}\|^2}}\right) 
		\\ \label{eq:rate1}
		&= \log_2 \Bigg(1+\frac{\alpha \rho \|{\bf h}_{\rm b}\|^4}{{\rho(1-\alpha)} \sum_{i \in \mathcal{L}}{|h_{{\rm b},i}|^4}+ {\|{\bf h}_{\rm b}\|^2}}\Bigg) 
	\end{align}
	\normalsize
	where $\mathcal{L}$ is the index set of non-zero complex gains of ${\bf h}_{\rm b}$ with $|\mathcal{L}| = L$, and ${h_{{\rm b},i}}$ is the $i$th element of ${\bf h}_{\rm b}$.
	Since $\|{\bf h}_{\rm b}\| = \sqrt{\gamma}$, maximizing $\mathcal{R}({\bf h}_{\rm b})$ in \eqref{eq:rate1} is equivalent to
	\begin{align}
		\label{eq:rate1_pf1}
		\min_{{\bf h}_{\rm b}}  \sum_{i \in \mathcal{L}}{|h_{{\rm b},i}|^4}	\quad \text{s.t. } \|{\bf h}_{\rm b}\|^2 = \gamma.
	\end{align}
	This can be solved by using Karush-Kuhn-Tucker conditions.
	Let $x_i = |h_{{\rm b},i}|^2$ for $i = 1,\dots,N_r$. 
	The Lagrangian of \eqref{eq:rate1_pf1} is
	\begin{align}
		\label{eq:rate1_pf2}
		\mathfrak{L}({\bf x}, \mu) = \|{\bf x}\|^2 + \mu \bigg(\sum_{i\in \mathcal{L}} x_i - \gamma \bigg).
	\end{align}
	By taking a derivative of \eqref{eq:rate1_pf2} with respect to $x_i$ for $i \in \mathcal{L}$ and setting it to zero, we have	$x_i = -{\mu}/{2}$.
	Putting this to the constraint $\|{\bf h}_{\rm b}\|^2 = \gamma$, the Lagrangian multiplier $\mu$ becomes $\mu = -2\gamma/L$, which leads to 
	\begin{align}
		\label{eq:rate1_pf3}
		x_i = {\gamma}/{L}, \quad i \in \mathcal{L}.
	\end{align}
	Since $x_i = |h_{{\rm b},i}|^2$ denotes the power of the beamspace complex gains, and $ L = |\mathcal{L}|$ represents the number of propagation paths, the result \eqref{eq:rate1_pf3} indicates that the beamspace channel maximizes the achievable rate in \eqref{eq:rate1} when the channel power $\|{\bf h}_{\rm b}\|^2 = \gamma$ is evenly spread to the $L$ beamspace complex gains for the single user case. 
	
	Accordingly, with $|h_{{\rm b},i}^\star|^2 = \gamma/L$ for $i \in \mathcal{L}$, the achievable rate in \eqref{eq:rate1} becomes 
	\begin{align}
		\label{eq:rate1_pf4}
		\mathcal{R}({\bf h}_{\rm b}^\star) 
		= \log_2 \left(1+\frac{\alpha \rho}{{\rho(1-\alpha)}/{L} + {1}/{\gamma}}\right).
	\end{align}
	Note that $\mathcal{R}({\bf h}_{\rm b}^\star)$ in \eqref{eq:rate1_pf4} increases as $L$ increases.
	Thus, a selected beamspace channel ${\bf h}_{\rm b}$ needs to have the largest number of propagation paths with equal power spread across the beamspace complex gains to maximize the rate.
	\end{proof}
	
Proposition \ref{pr:rate1} shows that the achievable rate is related not only to the channel magnitude $\|{\bf h}_{\rm b}\|$ but also to the channel structure in the beamspace for single user selection.
This reveals the difference from the conventional single user selection that only requires the largest channel magnitude to achieve a maximum rate. 
We further show that the maximum achievable rate for the single user case converges to a finite rate even with the infinite channel magnitude due to quantization errors if the number of propagation paths $L$ is limited.

\begin{corollary}
	If the number of channel propagation paths $L$ is finite, the maximum achievable rate with single user selection is limited to be finite with $\|{\bf h}_{{\rm b}}\| \to \infty$ and it converges to
	\begin{align}
		\label{eq:rate1_inf}
		\mathcal{R}({\bf h}_{\rm b}^\star) 
		\to \log_2 \left(1+{\alpha L}/{(1-\alpha)}\right).
	\end{align}
\end{corollary}
\begin{proof}
	By increasing the beamspace channel magnitude to infinity ($\gamma \to \infty$), the maximum achievable rate of the single user selection \eqref{eq:rate1_pf4} converges to \eqref{eq:rate1_inf}.
\end{proof}
Unlike the conventional single user systems in which the achievable rate increases to infinity as channel gain increases, the quantization error ($\alpha < 1$) limits the maximum rate to remain finite.  

To extend the derived result for the single user selection to the multi-user selection, we consider the case in which $L N_s \leq N_r$.
Indeed, this assumption is relevant to mmWave channels where the number of channel propagation paths is considered to be small compared to the number of antennas $L \ll N_r$.
Proposition \ref{pr:rate_multi} shows the structural criteria of channels to maximize the sum rate for multi-user selection.
\begin{proposition}
	\label{pr:rate_multi}
	For the case where $L N_s\leq N_r$ and $\|{\bf h}_{{\rm b},k}\| = \sqrt{\gamma} $, $\forall k$, selecting a set of users $\mathcal{S}$ that satisfies the following channel characteristics maximizes the uplink sum rate.
	\begin{enumerate}[(i)]
		\item Unique AoAs at the receiver for the channel propagation paths of each scheduled user.
		\begin{align}
			\label{eq:unique AoAs}
			\mathcal{L}_{\mathcal{S}(k)} \cap \mathcal{L}_{\mathcal{S}(k')} = \phi \text{ if } k \neq k'.
		\end{align} 
		\item Equal power spread across the beamspace complex gains within each user channel.
		\begin{align}
			\label{eq:Equal spread}
			|h_{{\rm b},i,\mathcal{S}(k)}| = \sqrt{\gamma/L} \text{ for } i \in \mathcal{L}_{\mathcal{S}(k)}
		\end{align} 
	\end{enumerate}
\end{proposition}

\begin{proof}
	The achievable rate in \eqref{eq:rate} can be expressed as
	
	\vspace{-1em}
	\small
	\begin{align}
		\label{eq:rate_multi_pf}
		\mathcal{R}_k({\bf H}_{\rm b}) 
		= \log_2 \left(1+\frac{\alpha \rho}{\rho\beta{\bf w}_{{\rm zf},k}^H {\rm diag}\Big({\bf H}_{\rm b} {\bf H}_{\rm b}^H\Big){\bf w}_{{\rm zf},k} + \|{\bf w}_{{\rm zf},k}\|^2}\right). 
	\end{align}
	\normalsize
	To prove Proposition \ref{pr:rate_multi}, we take a two-stage maximization approach and show that the conditions for the maximization are sufficient to solve \eqref{eq:opt_prob} when $LN_s \leq N_r$. 
	In the first stage, we focus on minimizing $\|{\bf w}_{{\rm zf},k}\|^2$ in \eqref{eq:rate_multi_pf}. 
	The zero-forcing combiner ${\bf w}_{{\rm zf},k}$ which satisfies ${\bf w}_{{\rm zf},k}^H{\bf h}_{{\rm b},k} = 1$ with minimum norm is ${\bf w}_{{\rm zf},k} = {\bf h}_{{\rm b},k}/\|{\bf h}_{{\rm b},k}\|^2$.
	When user channels are orthogonal, ${\bf h}_{{\rm b},k} \! \perp\! {\bf h}_{{\rm b},k'}$ for $ k \neq k'$, the zero-forcing combiner ${\bf w}_{{\rm zf},k} = {\bf h}_{{\rm b},k}/\|{\bf h}_{{\rm b},k}\|^2$ further satisfies ${\bf w}_{{\rm zf},k}^H{\bf h}_{{\rm b},k'}= 0$ for $k \neq k'$.   
	Therefore, with the orthogonality condition,  ${\bf w}_{{\rm zf},k} = {\bf h}_{{\rm b},k}/\|{\bf h}_{{\rm b},k}\|^2$ minimizes $\|{\bf w}_{{\rm zf},k}\|^2$.
	
	In the second stage, we minimize \eqref{eq:rate_multi_pf} given the orthogonality condition from the first stage as follows:
	
	\vspace{-1em}
	\small
	\begin{align}
		\nonumber
		&\mathcal{R}_k({\bf H}_{\rm b}|{\bf h}_{{\rm b},k}\! \perp \! {\bf h}_{{\rm b},k'}, k \neq k') 
		\\ \label{eq:rate_multi_pf0} 
		& \stackrel{(a)} = \log_2 \Biggl(1+\frac{\alpha \rho {\|{\bf h}_{{\rm b},k}\|^4} }{\rho\beta {\bf h}^H_{{\rm b},k}{\rm diag}\Big({\bf H}_{\rm b} {\bf H}_{\rm b}^H\Big) {\bf h}_{{\rm b},k}+ \|{\bf h}_{{\rm b},k}\|^2}\Biggr) 
		\\ \nonumber
		& = \log_2 \left (1+\frac{\alpha \rho \gamma^2}{\rho \beta  \underset {i\in \mathcal{L}_{k}}{\sum} |h_{{\rm b},i,k}|^2\bigg(|h_{{\rm b},i,k}|^2 + \underset{u\neq k}{\sum} |h_{{\rm b},i,u}|^2 \bigg ) + \gamma}\right)
		\\	\label{eq:rate_multi_pf1}
		&\stackrel{(b)} \leq \log_2 \left(1+\frac{\alpha \rho \gamma^2}{\rho \beta  \sum_{i\in \mathcal{L}_{k}}|h_{{\rm b},i,k}|^4 + \gamma}\right).
		\\	\label{eq:rate_multi_pf2}
		&\stackrel{(c)} \leq \log_2 \Bigg(1+\frac{\alpha \rho}{ \rho \beta /L + 1/\gamma}\Bigg).
	\end{align}
	\normalsize
	The equality in (a) comes from the channel orthogonality condition $({\bf h}_{{\rm b},k}\! \perp \!{\bf h}_{{\rm b},k'})$ which leads to ${\bf w}_{{\rm zf},k} = {{\bf h}_{{\rm b},k}}/{\|{\bf h}_{{\rm b},k}\|^2}$. 
	The equality in (b) holds if and only if $|h_{{\rm b},i,u}| = 0$, $\forall \, u \neq k$ and $i \in \mathcal{L}_k$; i.e, each user has channel propagation paths with unique AoAs.
	Note that \eqref{eq:rate_multi_pf1} is equivalent to the rate of single user selection in \eqref{eq:rate1} due to the channel orthogonality and unique AoA condition.
	Consequently, inequality (c) comes from the fact that \eqref{eq:rate_multi_pf1} is maximized when $|h_{{\rm b},i,k}| = \sqrt{\gamma/L}$ for $i \in \mathcal{L}_k$ as shown in Proposition \ref{pr:rate1}. 
	The upper bound in \eqref{eq:rate_multi_pf2} is equivalent to the maximum achievable rate for the single user case in \eqref{eq:rate1_pf4}; i.e., \eqref{eq:rate_multi_pf2} is the maximum achievable rate of each user for the multi-user case in the considered system, which also maximizes the sum rate in \eqref{eq:opt_prob}.
	
	Accordingly, the orthogonality condition, the unique AoA condition, and the equal power spread condition are sufficient to maximize the sum rate.
	Indeed, the unique AoA condition implies the orthogonality condition from the first stage.
	Therefore, the beamspace channel matrix ${\bf H}_{\rm b}$ only needs to satisfy the unique AoA and equal power spread conditions to maximize the uplink sum rate.
	This completes the proof.
\end{proof}

Intuitively, the structural criteria in Proposition \ref{pr:rate_multi} imply that we need to select users by minimizing the quantization noise to maximize the sum rate.
Conventional user scheduling approaches ignore such structural criteria of the channel in the beamspace as perfect quantization is considered, and the channel orthogonality and magnitude are mainly considered in the conventional user scheduling methods, which is not suitable for the systems with low-resolution ADCs.


\section{Proposed Algorithm}
\label{sec:algorithm}

\begin{algorithm}[!t]
\caption{Channel Structure-based Scheduling (CSS)}
\label{algo:CSS}
\vspace{.3 em}
\begin{enumerate}
	\item BS initializes $\mathcal{U}_1 = \{1,\dots, N_u\}$, $\mathcal{S} = \phi$, and $i = 1$.
	\item BS stores $N_b \geq L$ dominant beam indices of each user channel ${\bf h}_{{\rm b}, k}$ in $\mathcal{B}_k$, $\forall k$.
	\item Using SINR in \eqref{eq:sinr}, BS selects $i$th user as
	\begin{align}
		\nonumber
		\mathcal{S}(i) = \argmax_{k \in \mathcal{U}_i} {\rm SINR}_k\big([{\bf H}_{\rm b}(\mathcal{ S}), {\bf h}_{{\rm b},k}]\big)
	\end{align}
	and updates $\mathcal{U}_{i} = \mathcal{U}_i\setminus\{\mathcal{S}(i) \}$ and $\mathcal{S}= \mathcal{S} \cup \{\mathcal{S}(i) \}$.
	\item To calculate component of ${\bf h}_{{\rm b}, \mathcal{S}(i)}$ that is orthogonal to subspace ${\rm span}\{{\bf f}_{\mathcal{S}(1)},\dots,{\bf f}_{\mathcal{S}(i-1)}\}$, computes
	\begin{align}
		\nonumber
		{\bf f}_{\mathcal{S}(i) } &= {\bf h}_{{\rm b},\mathcal{S}(i) } - \sum_{j = 1}^{i-1}  \frac{{\bf f}_{\mathcal{S}(j)}^H{\bf h}_{{\rm b},\mathcal{S}(i) }}{\|{\bf f}_{\mathcal{S}(j)}\|^2}{\bf f}_{\mathcal{S}(j)}\\ \nonumber
		&=\bigg({\bf I} - \sum_{j = 1}^{i-1} \frac{{\bf f}_{\mathcal{S}(j)}{\bf f}_{\mathcal{S}(j)}^H}{\|{\bf f}_{\mathcal{S}(j)}\|^2} \bigg){\bf h}_{{\rm b},\mathcal{S}(i)}
	\end{align}
	\item To impose orthogonality among selected users by using results from step 2 and 4, BS further updates
	\begin{align}
		\nonumber
		\mathcal{U}_{i+1} = \bigg\{k\in \mathcal{U}_i\ |&\ \frac{|{\bf f}_{\mathcal{S}(i)}^H{\bf h}_{{\rm b},k}|}{\|{\bf f}_{\mathcal{S}(i)}\| \|{\bf h}_{{\rm b},k}\|}< \epsilon,\\ \label{eq:semi-orthogonal}
		 &\quad \quad \quad  |\mathcal{B}_{\mathcal{S}(i)} \cap \mathcal{B}_k | \leq N_{\rm OL}  \bigg\}
	\end{align}
	\item If $i \leq N_s$ and $\mathcal{U}_{i+1} \neq \phi $, update selection stage $i = i + 1$ and go to step 3. 
	Otherwise, algorithm finishes.
\end{enumerate}
\end{algorithm}

In this section, we propose a user scheduling algorithm with low complexity by using the derived structural criteria with respect to the beamspace channel in Proposition \ref{pr:rate_multi}.
Leveraging the channel orthogonality condition, the proposed algorithm first filters the user set with a semi-orthogonality condition as in \cite{yoo2006optimality}, which leaves only the users whose beamspace channels closely satisfy the equality in \eqref{eq:rate_multi_pf0}. 
Exploiting the unique AoA condition \eqref{eq:unique AoAs}, the algorithm, then, enforces an additional spatial orthogonality in the beamspace as in \cite{lee2016performance} to the filtered set.
These filtering steps reduce the size of the user set, offering semi-orthogonality between the selected users and remaining users.
This semi-orthogonality leads to ${\bf w}_{{\rm zf},k} \approx {\bf h}_{{\rm b},k}/\|{\bf h}_{{\rm b},k}\|^2$ when selecting a user.

Using ${\bf w}_{{\rm zf},k} \approx {\bf h}_{{\rm b},k}/\|{\bf h}_{{\rm b},k}\|^2$, the signal-to-interference-plus-noise ratio (SINR) of user $k$ in the candidate set can be approximated as
\begin{align}
	\label{eq:sinr}
	{\rm SINR}_k({\bf H}_{\rm b}) \approx \frac{\alpha \rho \|{\bf h}_{{\rm b},k}\|^4}{(1-\alpha){\bf h}_{{\rm b},k}^H\, {\bf D}({\bf H}_{\rm b})\, {\bf h}_{{\rm b},k}}
\end{align}
where ${\bf D}({\bf H}_{\rm b}) = {\rm diag}(\rho{\bf H}_{\rm b} {\bf H}_{\rm b}^H + \frac{1}{1-\alpha}{\bf I}_{N_{\rm RF}})$.
As a selection measure, the algorithm adopts the approximated SINR \eqref{eq:sinr} to incorporate the structural criteria of channels in Proposition \ref{pr:rate_multi} with the channel magnitude and orthogonality.
Using the approximated SINR \eqref{eq:sinr} greatly reduces complexity by avoiding the matrix inversion for computing zero-forcing combiners ${\bf W}_{\rm zf}$.
The proposed algorithm, which we call channel structure-based scheduling (CSS), is described in Algorithm \ref{algo:CSS}.
Note that due to the quantized beamforming angle, there exist beamforming offsets that result in more than $L$ dominant beams in the channels ${\bf h}_{{\rm b},k}$.  
Consequently, the algorithm stores $N_b \geq L$ dominant beam indices as in step 2 in Algorithm \ref{algo:CSS}.

\begin{algorithm}[!t]
\caption{Greedy User Scheduling}
\label{algo:greedy}
\vspace{.3 em}
\begin{enumerate}
	\item BS initializes $\mathcal{T}_1 = \{1,\dots,N_u\}$, $\mathcal{S}_G = \phi$, and $i = 1$.
	\item BS selects a user who maximizes sum rate as
	\begin{align}
		\mathcal{S}_G(i) = \argmax_{k \in \mathcal{T}_i} \sum_{j \in {\mathcal{S}_G \cup \{k\}}} \mathcal{R}_j\big([{\bf H}_{\rm b}(\mathcal{S}_G),{\bf h}_{{\rm b},k}]\big)
	\end{align}
	where $\mathcal{R}_j$ is given in \eqref{eq:rate}.
	\item Update $\mathcal{T}_{i+1} =\mathcal{T}_i\setminus \{\mathcal{S}_G(i)\}$, $\mathcal{S}_G = \mathcal{S}_G \cup \{\mathcal{S}_G(i) \}$, and $i = i + 1$, and go to step 2 until select $N_s$ users.
\end{enumerate}
\end{algorithm}

To provide a reference for sum rate performance, we propose a greedy algorithm that schedules the user who achieves the highest sum rate in each iteration.
The greedy scheduling algorithm, given as Algorithm \ref{algo:greedy}, offers sub-optimal performance with high complexity.
In each iteration, the greedy approach uses the achievable rate in \eqref{eq:rate} as a selection measure, and thus, computes the exact SINR with zero-forcing combiners for each scheduled user.
Accordingly, the greedy method carries the burden of computing a matrix inversion in each iteration.
The method computes the achievable rate $|\mathcal{T}_i| \times i$ times and compares the resulting $|\mathcal{T}_i|$ sum rates for the $i$th selection, whereas the CSS method only computes the approximated SINR $|\mathcal{U}_i|$ times and compares $|\mathcal{U}_i|$ SINRs.
Moreover, unlike the greedy algorithm, the CSS algorithm reduces the size of the remaining user candidate set $\mathcal{U}_i$ after each selection by enforcing the semi-orthogonality conditions in \eqref{eq:semi-orthogonal}. 
This size reduction leads to $|\mathcal{U}_i| < |\mathcal{T}_i|$, which makes the CSS algorithm more computationally efficient than the greedy algorithm in two cases: the total number of candidate users $N_u$ becomes larger, and more users need to be selected.

\section{Simulation Results}
\label{sec:simulation}

\begin{figure}[!t]\centering
\includegraphics[scale = 0.15]{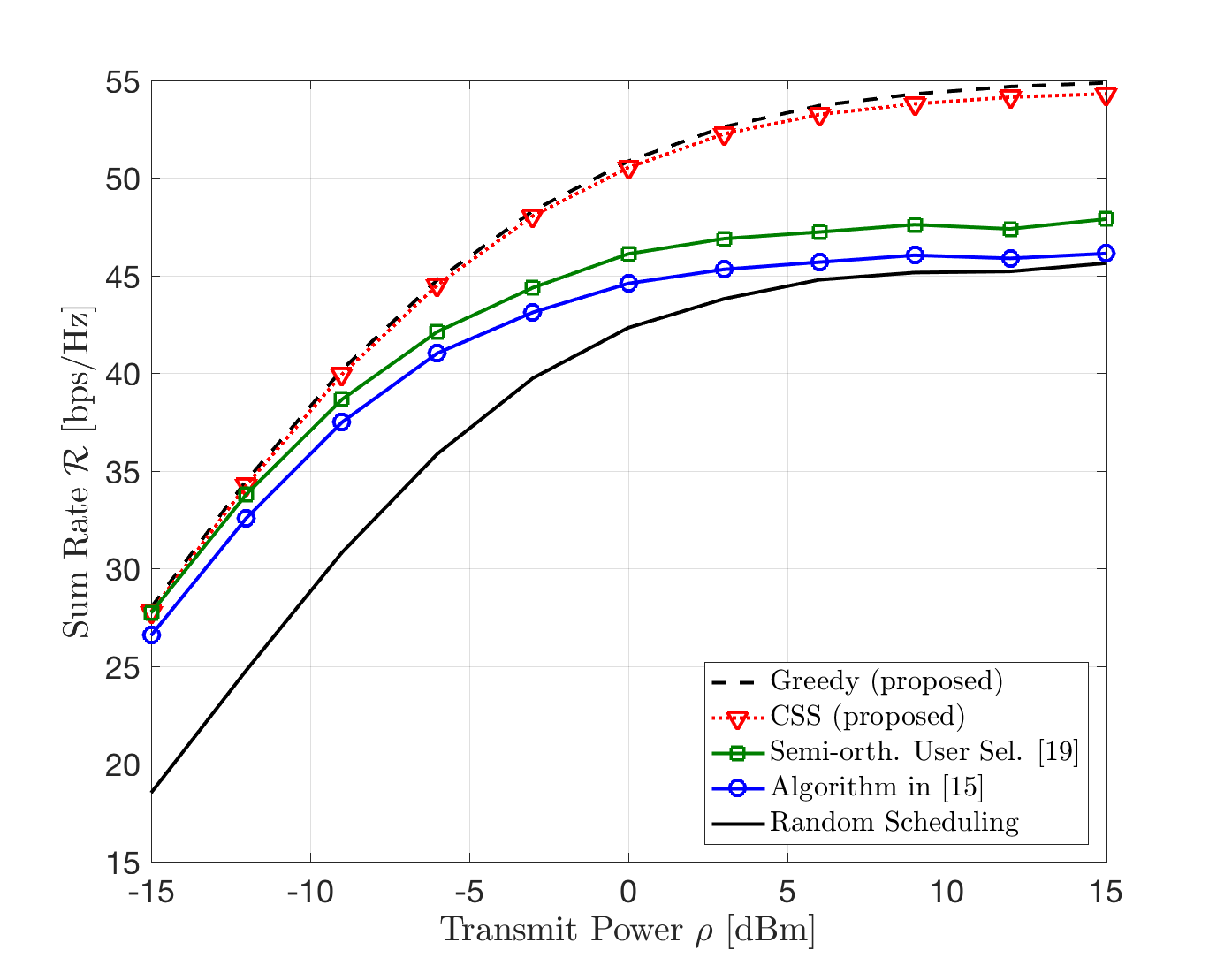}
\caption{The uplink sum rate with respect to the transmit power $\rho$ with $N_r = 128$, $N_u = 200$, $N_s = 10$, $L = 4$, $N_b = 8$, $b = 2$, and $N_{\rm OL} = 3$.} 
\label{fig:rate_power}
\vspace{-1em}
\end{figure}

\begin{figure}[!t]\centering
\includegraphics[scale = 0.15]{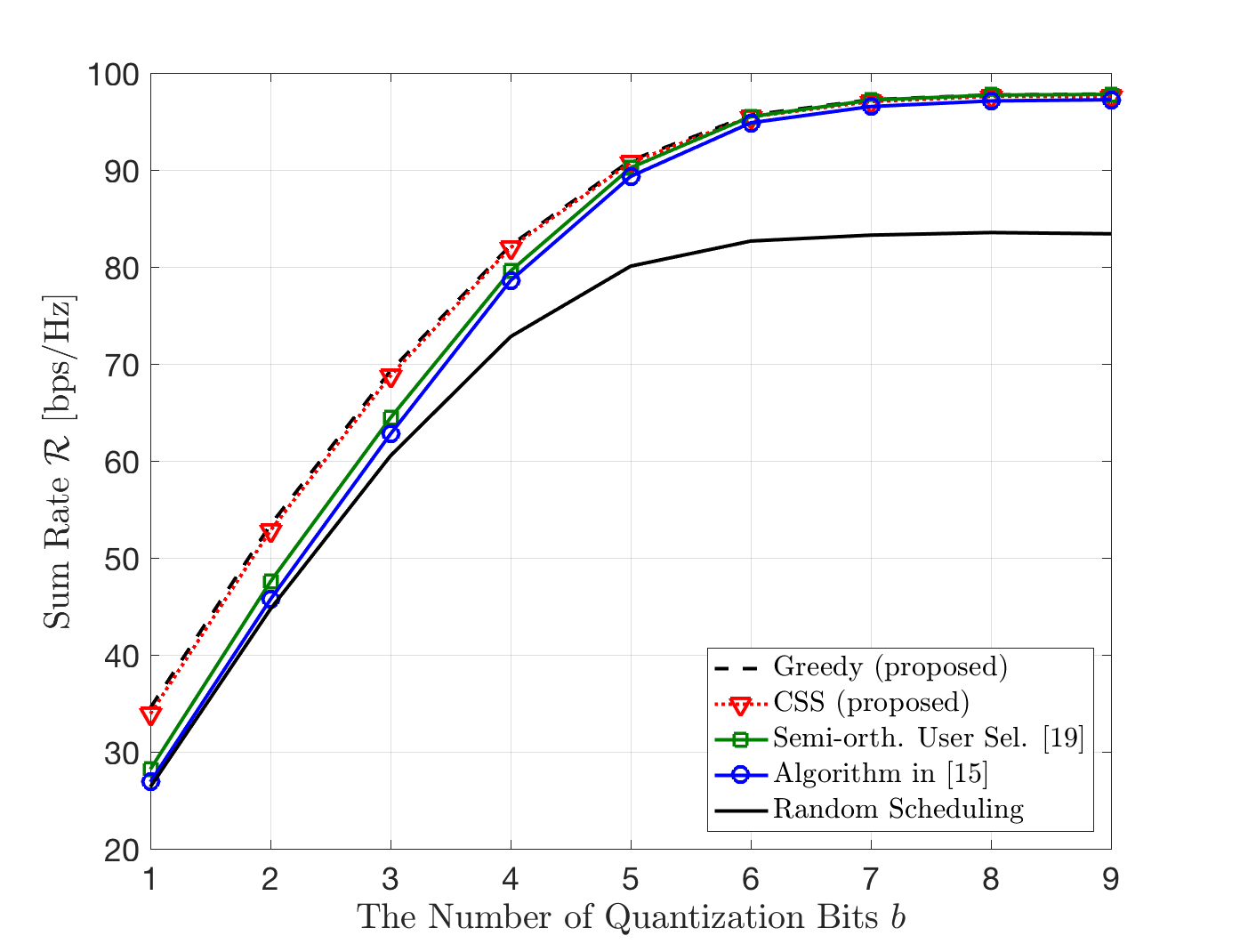}
\caption{The uplink sum rate with respect to the number of quantization bits $b$ with $N_r = 128$, $N_u = 200$, $N_s = 10$, $L = 4$, $N_b = 8$, and $\rho = 5$ dB.} 
\label{fig:rate_bit}
\vspace{-1em}
\end{figure}

In this section, we numerically evaluate the proposed CSS algorithm. 
The Matlab code is available at \cite{choi2017usersoftware}.
We consider $N_r = 128$ BS antennas, $N_u = 200$ users, $L = 4$ channel propagation paths for each user, and $N_b = 2\times L$ beam indices to store at the BS.
We assume that the BS schedules $N_s = 10$ users to serve at each transmission.
In addition to the proposed CSS and greedy algorithms, we simulate the algorithm in \cite{lee2016performance} and semi-orthogonal user scheduling (SUS) algorithm in \cite{yoo2006optimality}.
To provide a reference for a performance lower bound, a random scheduling case is also included.
In simulations, the channel model in \eqref{eq:channel_geo} is used 
without imposing the constraint of $\|{\bf h}_{{\rm b},k}\| = \sqrt{\gamma}$, $\forall k$.
Orthogonality parameters such as $\epsilon$ and $N_{\rm OL}$ are optimally chosen unless mentioned otherwise.

Fig. \ref{fig:rate_power} shows the uplink sum rate with respect to the transmit power $\rho$.
We consider the number of quantization bits $b$ and the number of beam-overlap constraint $N_{\rm OL}$ to be $b =2$ and $N_{\rm OL}= 3$, respectively.
Notably, the CSS algorithm almost achieves the sum rate of the greedy algorithm with lower complexity. 
The gap between the CSS and the SUS algorithms increases as the transmit power increases because the quantization noise power becomes dominant compared to the additive white Gaussian noise power in the high SNR regime.
Accordingly, the gain from considering the structural criteria becomes larger as the SNR increases.
In the high SNR regime, the proposed CSS algorithm provides $22 \%$ sum rate increase compared to the random scheduling, whereas the algorithm in \cite{lee2016performance} and the SUS algorithm show marginal sum rate increase.

In Fig. \ref{fig:rate_bit}, we consider $\rho = 5$ dBm and increase $N_{\rm OL}$ from $2$ to $5$ as the number of quantization bits increases since the unique AoA condition becomes less critical as the quantization resolution increases;
$N_{\rm OL} =  2, 3, 4,$ and $5$ for $b \in \{1,2,3\}$, $b \in \{4, 5\}$, $b \in \{6, 7\}$, and $b \in \{8, 9\}$, respectively.
The CSS algorithm also attains the sum rate of the greedy algorithm with lower complexity and outperforms the algorithm in \cite{lee2016performance} and SUS algorithm in the low-resolution regime. 
Note that the sum rate of the algorithm in \cite{lee2016performance} and the SUS algorithm converges to that of the CSS and greedy algorithm as the number of quantization bits increases.
This convergence corresponds to the findings that the structural criteria with respect to  the channel in the beamspace, in addition to the channel magnitude and orthogonality, needs to be considered in user scheduling under coarse quantization.
Accordingly, in the low-resolution regime, there exists a noticeable sum rate gap between the proposed algorithms and the other algorithms.

\begin{figure}[!t]\centering
\includegraphics[scale = 0.15]{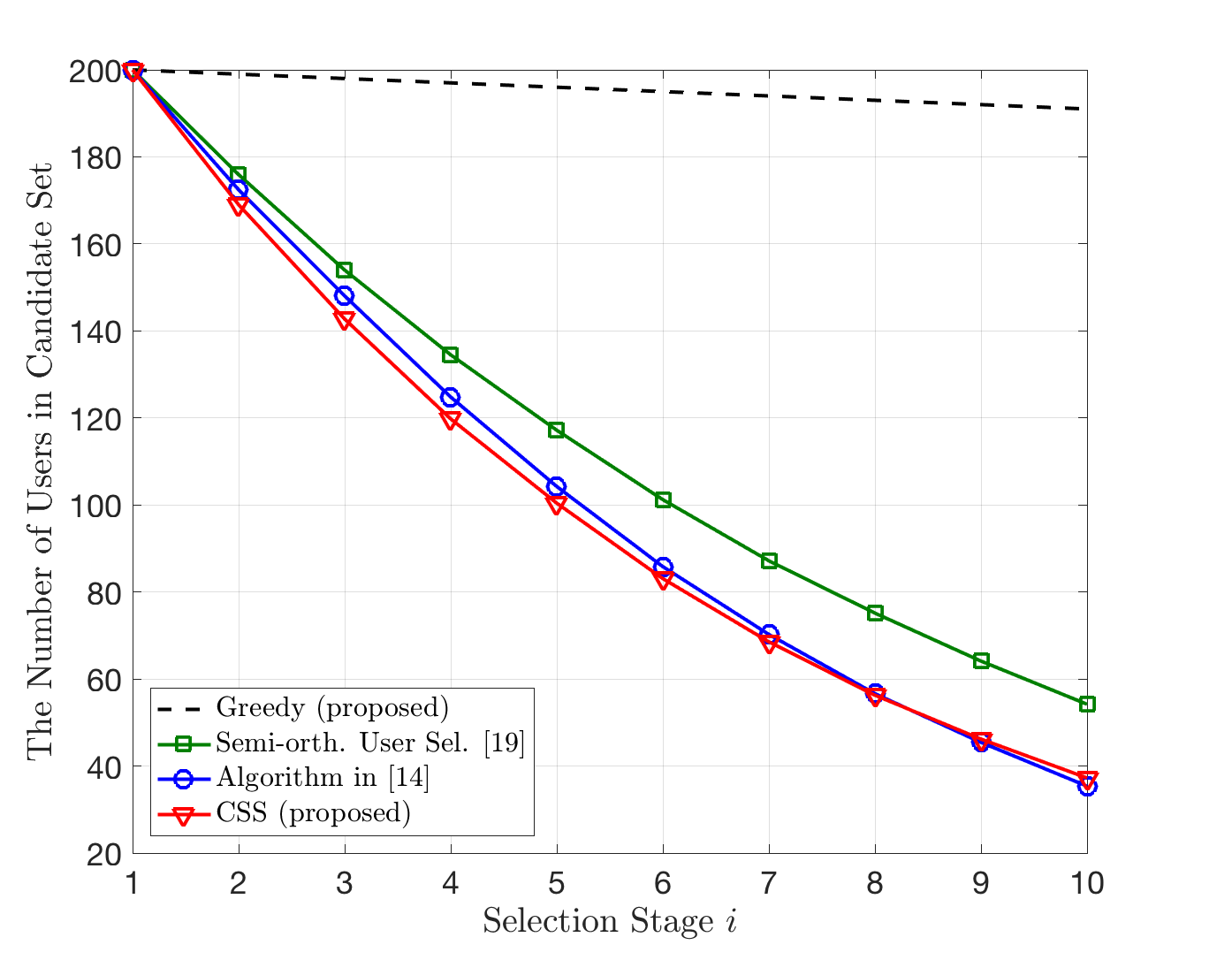}
\caption{The number of users in the candidate set at each selection stage $i$ for $\rho = 6$ dBm transmit power with $N_r = 128$, $N_u = 200$, $N_s = 10$, $L = 4$, $N_b = 8$, $b = 2$, and $N_{\rm OL} = 3$.} 
\label{fig:setsize}
\vspace{-1em}
\end{figure}

The number of remaining users in the candidate set at each selection stage $i$ is shown in Fig. \ref{fig:setsize}.
We consider the case of $\rho = 6$ dBm with $b = 2$ and $N_{\rm OL} = 3$, which corresponds to the $\rho = 6$ dBm point in Fig. \ref{fig:rate_power}.
The algorithms other than greedy scheduling show large decrease in the candidate set size as the algorithms adopt orthogonality conditions to filter the user.
In particular, the proposed CSS algorithm provides the largest decrease in the set size which is similar to the algorithm in \cite{lee2016performance} while achieving the sum rate comparable to that of the greedy scheduling.
Consequently, the simulation results reveal the efficiency of using the proposed CSS algorithm which accomplishes high performance with low complexity regarding the number of users in the candidate set.

\section{Conclusion}
\label{sec:conclusion}

In this paper, we investigated user scheduling for mmWave MIMO systems with low-resolution ADCs.
We discovered that in addition to channel magnitude and beamspace orthogonality, the structural characteristic of channels in the beamspace is indispensable in user scheduling with low-resolution ADCs due to the quantization error.
To maximize the achievable rate with respect to the channel structure, the channels of the scheduled users need to have $(1)$ as many propagation paths as possible with unique AoAs to give spatial orthogonality in the beamspace, and $(2)$ even power distribution in the beamspace to reduce the quantization error.
Leveraging these structural criteria, we proposed the scheduling algorithm with low complexity. 
The simulation results demonstrated that the proposed algorithm outperforms the conventional scheduling algorithms, providing a larger increase of the uplink achievable rate for mmWave systems with low-resolution ADCs.
{\color{black}Future work could consider imperfect channel information.}

\bibliographystyle{IEEEtran}
\bibliography{icc18.bib}
\end{document}